\documentclass{arXiv}

\usepackage[dvips]{graphicx}
\usepackage{bbm}

\newcommand{\bela}[1]{\begin{equation}\label{#1}}
\newcommand{\ela}{\end{equation}}
\newcommand{\bear}[1]{\begin{array}{#1}}
\newcommand{\ear}{\end{array}}

\renewcommand{\P}{\mbox{\boldmath $P$}}

\renewcommand{\Psi}{\mbox{\boldmath $\psi$}}

\newcommand{\as}{\\[.6em]}

\renewcommand{\i}{\mbox{\rm i}}

\newcommand{\tr}{\,\mbox{tr}\,}

\newcommand{\va}[1]{{\tilde{#1}}}

\newtheorem{definition}[theorem]{Definition}
\newtheorem{observation}[theorem]{Observation}
\newtheorem{convention}[theorem]{Convention}

\begin{document}

\title[Generalized Clifford configurations and the qdSKP equation]{A novel 
generalization of Clifford's 
  classical point-circle configuration. Geometric interpretation of the 
  quaternionic discrete Schwarzian KP
  equation}

\author[W. K.\ Schief and B. G.\ Konopelchenko]{W. K. Schief$\,^{1,2}$ and 
B.\ G.\ Konopelchenko\,$^3$}

\affiliation{$^1$ Institut f\"ur Mathematik, Technische Universit\"at Berlin, 
Stra\ss e des 17. Juni 136, D-10623 Berlin, Germany \\[2mm]
$^2$ Australian Research Council Centre of Excellence for Mathematics and 
Statistics of Complex Systems, School of Mathematics, The University of New 
South Wales, Sydney, NSW 2052, Australia\\[2mm]
$^3$ Dipartimento di Fisica, Universit\`a del Salento and INFN, Sezione di
Lecce, 73100~Lecce, Italy}

\label{firstpage}

\maketitle

\begin{abstract}{Clifford configurations; quaternions; multi-ratios; 
integrable systems; discrete differential geometry}
 The algebraic and geometric properties of a novel 
 generalization of Clifford's classical $\mathcal{C}_4$ point-circle 
 configuration are analysed. A connection with the integrable quaternionic 
 discrete Schwarzian Kadomtsev-Petviashvili equation is revealed.
\end{abstract}

\section{Introduction}

In his seminal paper {\it Clifford's chain and its analogues in relation 
to the higher polytopes}, Longuet-Higgins (1972)
asserts that "a chain of theorems
... has exerted a peculiar fascination for mathematicians since its
discovery by Clifford in 1871". Indeed, various generalizations and analogues 
in higher dimension of Clifford's point-circle configurations $\mathcal{C}_n$
(Clifford 1871) associated with such luminaries as 
de Longechamps (1877), Cox (1891), 
Grace (1898), Brown (1954), Coxeter (1956) and 
Longuet-Higgins (1972) have been recorded and analysed in 
detail. The 
original celebrated chain of `circle theorems' may be stated as follows:

\medskip
\noindent
Given four straight lines on a plane, the four circumcircles of the four
triangles so formed are concurrent in a point $Q_4$, say (cf.\ figure
\ref{menelaus}).

\medskip
\noindent
Given five lines on a plane, by omitting each line in turn, we obtain
five corresponding points $Q_4$ and these lie on a circle $C_5$, say.

\medskip
\noindent
Given six lines on a plane, we obtain six corresponding circles $C_5$
and these are concurrent in a point $Q_6$.

\medskip
\noindent
Etc.

\medskip
\noindent
Generally, given $n$ coplanar lines, we obtain $n$ corresponding circles 
$C_{n-1}$ which are concurrent in a point $Q_n$ or $n$ points $Q_{n-1}$ 
which lie on a circle $C_n$ depending on whether $n$ is even or odd 
respectively.  

\medskip
\noindent
Finally, application of an inversion with respect to a 
generic point on the plane leads to a complete and symmetric configuration of
$2^{n-1}$ points and $2^{n-1}$ circles with $n$ points on every circle
and $n$ circles through every point.

\begin{figure}
\centerline{\includegraphics[width=0.5\textwidth]{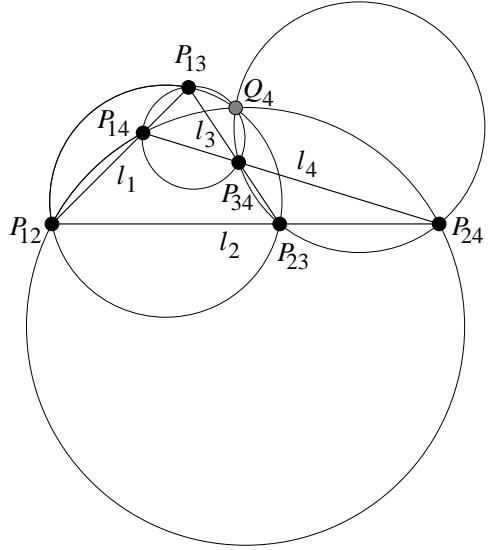}}
\caption{A `Menelaus configuration'}
\label{menelaus}
\end{figure}

\medskip
\noindent
The first theorem associated with four straight lines immediately demonstrates
that there exists a close connection between $\mathcal{C}_4$ Clifford 
configurations and the ancient Theorem of Menelaus (Pedoe 1970). The latter 
states that three points $P_{14},P_{24},P_{34}$ on the (extended) edges of
a triangle with vertices $P_{12},P_{23},P_{13}$ as displayed in 
figure~\ref{menelaus} are collinear if and only if 
the condition
\bela{I1}
  \frac{\overline{P_{12}P}\!_{24}\,\overline{P_{23}P}\!_{34}\,
  \overline{P_{13}P}\!_{14}}{
  \overline{P_{24}P}\!_{23}\,\overline{P_{34}P}\!_{13}\,
  \overline{P_{14}P}\!_{12}} = -1
\ela
for the associated directed lengths is satisfied. Accordingly, 
the points of intersection of the four lines $l_1,l_2,l_3,l_4$ in Clifford's
first theorem (see figure \ref{menelaus}) obey condition (\ref{I1}).
If the plane is identified with the complex plane then condition (\ref{I1}) 
constitutes a multi-ratio condition for the complex numbers $P_{ik}$ which we
denote by (cf.~\S 2)
\bela{I2}
  M(P_{14},P_{12},P_{24},P_{23},P_{34},P_{13}) = -1
\ela
(modulo a trivial cyclic permutation of the arguments).
The latter is evidently invariant under the group of inversive 
transformations (Brannan \textit{et al.} 1999) 
and hence the points $P_{ik}$ of a $\mathcal{C}_4$ Clifford 
configuration likewise satisfy the multi-ratio condition (\ref{I2}). In fact,
it has been pointed out in Konopelchenko \& Schief (2002) that the latter 
constitutes a defining property of $\mathcal{C}_4$ Clifford configurations.

In this paper, we investigate in detail the geometric and algebraic properties
of a novel generalization of Clifford's $\mathcal{C}_4$ configuration which has
been discovered via the theory of integrable systems (soliton theory)
(Ablowitz \& Segur 1981; Zakharov \textit{et al.} 1980). 
Thus, in Konopelchenko \& Schief (2002), it has been shown that
(\ref{I2}) interpreted as a lattice equation which is defined on the 
`octahedral' vertex configurations of a face-centred cubic (fcc) lattice 
(cf.\ \S 7) constitutes a Schwarzian
version (Dorfman \& Nijhoff 1991; Bogdanov \& Konopelchenko 1998) 
of the Hirota-Miwa equation, that is the discrete Kadomtsev-Petviashvili
(dKP) equation (Hirota 1981). The latter is regarded as a `master equation'
in soliton theory since it encodes the complete KP hierarchy of soliton 
equations via sophisticated continuum limits. The discrete Schwarzian KP
(dSKP) equation admits a natural multi-component analogue, namely 
the natural matrix generalization of the multi-ratio condition (\ref{I2}) 
interpreted as a lattice equation (Bogdanov \& Konopelchenko 1998). In the 
simplest case, the quaternionic dSKP (qdSKP) equation locally represents
a six-point relation for six quaternions $P_{ik}$, that is, a relation between
six points $P_{ik}$ in a four-dimensional Euclidean space~$\mathbbm{R}^4$ if the
standard identification $\mathbbm{R}^4\cong\mathbbm{H}$ with the algebra of
quaternions is made. Since the multi-ratio condition (\ref{I2}) encodes
$\mathcal{C}_4$ Clifford configurations, it is natural to inquire as to the
geometric significance of its quaternionic counterpart. It turns out that an
appropriate characterization of classical $\mathcal{C}_4$ Clifford 
configurations gives rise to natural analogues in a four-dimensional Euclidean 
space which are algebraically governed by the local qdSKP equation.

In the present context, the key property of classical $\mathcal{C}_4$ Clifford 
configurations turns out to be the Godt-Ziegenbein property which states that,
in a specific sense, the angles made by four oriented circles passing through a
point are the same for all eight points (Godt 1896; Ziegenbein 1941). This 
property is used in \S 3 to define octahedral point-circle 
configurations in $\mathbbm{R}^4$ of Clifford type. In \S 6, the existence 
of such generalized Clifford configurations is proven and it is demonstrated 
that these are indeed governed by the afore-mentioned quaternionic version of 
the multi-ratio condition (\ref{I2}). As a by-product, it is shown that, 
for any five generic points in $\mathbbm{R}^4$, there exists a pair of
associated generalized Clifford configurations which are related by reflection
in the hypersphere defined by the given five points. The final section is then
devoted to the geometry of the qdSKP (lattice) equation.

\section{The classical \boldmath $\mathcal{C}_4$ Clifford configuration}

We begin with the classical construction of $\mathcal{C}_4$ Clifford 
configurations. Thus, consider a point $P$ on the plane and four generic 
circles $S_1,S_2,S_3,S_4$ passing through $P$ as depicted in 
figure \ref{circle}. The
six additional points of intersection are labelled by $P_{12},P_{13},
P_{14},P_{23},P_{24},P_{34}$. Here, the indices on $P_{ik}$
correspond to those of the circles $S_i$ and $S_k$. Any three circles 
$S_i,S_k,S_l$ intersect at three points and therefore define a circle $S_{ikl}$
passing through these points. Clifford's circle theorem (Clifford 1871)
then states that the four circles $S_{123},S_{124},S_{134},S_{234}$ meet at a 
point $P_{1234}$. 
Even though Clifford configurations ($\mathcal{C}_n$) exist 
for any number of initial circles $S_1,\ldots,S_n$ passing through a point $P$,
for brevity, we here associate with the term `Clifford configuration' the case
$n=4$.  
\begin{figure}
\centerline{\includegraphics[width=\textwidth]{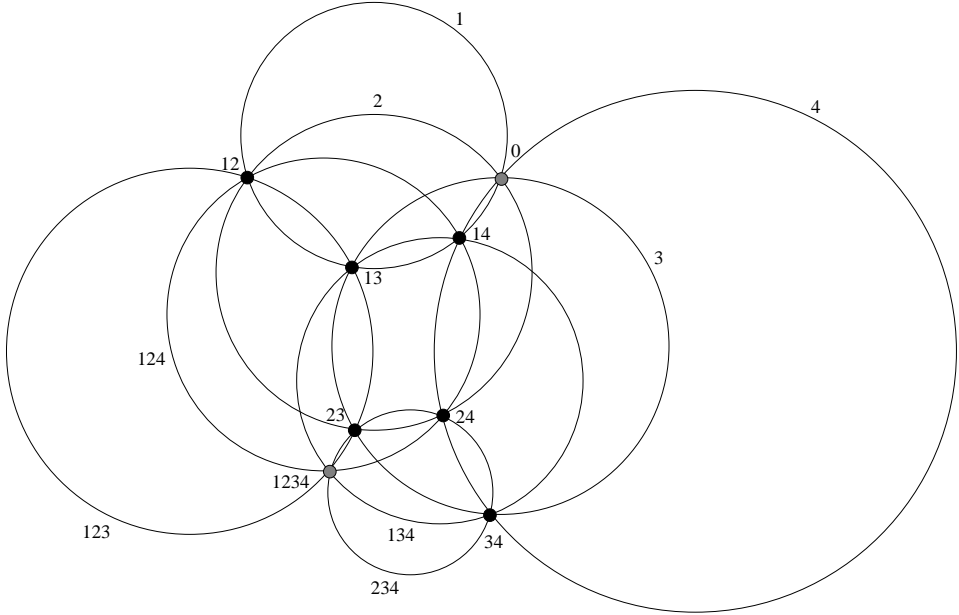}}
\caption{A classical $\mathcal{C}_4$ Clifford configuration}
\label{circle}
\end{figure}

In Konopelchenko \& Schief (2002), as an immediate consequence of the classical 
Theorem of Menelaus (Pedoe 1970), it has been demonstrated 
that any generic six 
points $P_{12},P_{13},P_{14}$, $P_{23},P_{24},P_{34}$ on the plane regarded 
as complex numbers belong to a Clifford configuration (with the 
above-mentioned combinatorics) if and only if they obey the multi-ratio 
condition
\bela{E1}
  M(P_{14},P_{12},P_{24},P_{23},P_{34},P_{13}) = -1,
\ela
where the multi-ratio of six complex numbers $P_1,\ldots,P_6$ is defined by
\bela{E2}
  M(P_1,P_2,P_3,P_4,P_5,P_6) = \frac{(P_1-P_2)(P_3-P_4)(P_5-P_6)}{
                                     (P_2-P_3)(P_4-P_5)(P_6-P_1)}.
\ela
In particular, this confirms the geometrically evident fact that 
any five generic points on the plane uniquely define a 
Clifford configuration. Indeed, for any given generic points
$P_{12},P_{13},P_{14},P_{23},P_{24}$, the sixth point $P_{34}$ is
determined by the linear equation (\ref{E1}). The latter point may lie `at
infinity', in which case the four circles passing through $P_{34}$ degenerate 
to straight lines. The justification of the ordering of the arguments in 
(\ref{E1}) is consigned to \S 6.

The preceding discussion indicates that one may think of a Clifford
configuration
as a configuration of six points  $P_{12},P_{13},P_{14},P_{23},P_{24},P_{34}$ 
and eight circles $S_1,S_2,S_3,S_4,S_{123},S_{124},S_{134},S_{234}$ which is
such that the four circles $S_1,S_2,S_3,S_4$ intersect
at a point $P$ or the four circles $S_{123},S_{124},S_{134},S_{234}$
intersect at a point $P_{1234}$. Clifford's circle theorem then guarantees the 
existence of the remaining point $P_{1234}$ or $P$ respectively. In this 
connection, it is noted that the eight points $P,\ldots,P_{1234}$ of a Clifford
configuration appear on an equal footing so that, at first sight, the above 
interpretation of a Clifford configuration does not seem to be natural. 
However, it turns out that it is precisely this point of view which allows for
a generalization of Clifford configurations in which, generically, the points 
$P$ and $P_{1234}$ do not exist.

A remarkable property of Clifford configurations is that the angles made by four
oriented circles passing through a point are the same for all eight points in a
sense to be specified in the following section. Therein, it is shown that this
Godt-Ziegenbein property (Godt 1896; Ziegenbein 1941) constitutes a 
defining property of Clifford configurations. In fact, it is sufficient to 
demand that the Godt-Ziegenbein property holds for the points 
$P_{12},P_{13},P_{14},P_{23},P_{24},P_{34}$. The latter observation serves as 
the basis for the definition of generalized Clifford configurations.  

\section{Octahedral point-circle configurations. Definitions and notation}

In the following, we are concerned with configurations in a 
four-dimensional Euclidean space $\mathbbm{R}^4$ consisting of six points and 
eight circles with three points on every circle and four circles through every 
point. More precisely (cf.\ figure~\ref{octahedron}):
\begin{definition}
  {\bf (Octahedral point-circle configurations)}
  A configuration of six points and eight circles in $\mathbbm{R}^4$ is termed
  an {\em octahedral point-circle configuration} if the combinatorics of the
  configuration is that of an octahedron, that is the points of the
  configuration correspond to the vertices of the octahedron while the circles
  correspond to the triangular faces.
\end{definition}
\begin{figure}
\centerline{\includegraphics[width=0.8\textwidth]{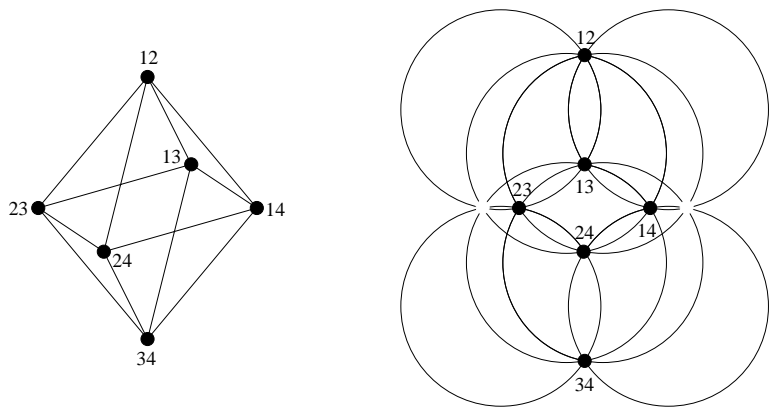}}
\caption{The combinatorics of an octahedral point-circle configuration}
\label{octahedron}
\end{figure}

In order to define octahedral point-circle configurations of Clifford type, it
is necessary to introduce a correspondence between circles which pass through
different points. To this end, we observe that any vertex of an octahedron may
be associated with its `opposite' counterpart, that is the vertex which is not
connected via an edge. Similarly, there exist four pairs of disconnected 
`opposite' faces. Thus, by virtue of the combinatorial correspondence
employed in the above definition, any point $P$ of an octahedral
point-circle configuration admits an `opposite' point $P^*$ and any circle $S$
is associated with an `opposite' circle $S^*$ (cf.\ figure \ref{opposite}).
\begin{figure}
\centerline{\includegraphics[width=0.8\textwidth]{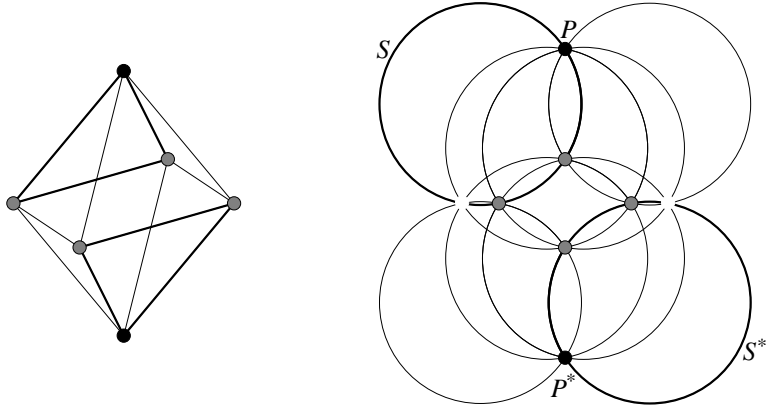}}
\caption{`Opposite' points and circles}
\label{opposite}
\end{figure}

\begin{definition}
  {\bf (Correspondence of circles)}
  Let $P,P_1,P_2$ and $S,S_1,S_2$ be points and circles of an octahedral
  point-circle configuration such that $S_1$ and $S_2$ intersect at $P_1$ and
  $P_2$ and $P$ lies on $S$. Then, the circles $S_1$ and $S_2$
  passing through $P_1$ are said to {\em correspond} to the circles $S_2$
  and $S_1$ respectively passing through $P_2$ and the circle $S$ passing
  through $P$ is said to {\em correspond} to the opposite circle $S^*$ passing
  through the opposite point $P^*$:
\bela{E3}
   (S_1,S_2;P_1,P_2)\leftrightarrow(S_2,S_1;P_2,P_1),\quad
   (S;P)\leftrightarrow(S^*;P^*). 
\ela
\end{definition}

Iterative application of the above correspondence principle immediately leads
to the following correspondence:

\begin{observation}
  Any circle $S_1$ of an octahedral point-circle configuration passing
  through a point $P_1$ admits five corresponding circles $S_2,\ldots, S_6$
  which pass through the remaining five points $P_2,\ldots,P_6$. Thus, there
  exists a unique correspondence between the six sets of four circles
  $S^{\mu}_k,\,\mu=1,2,3,4$ passing
  through the points $P_k$ of an octahedral point-circle configuration.
\end{observation}

\begin{convention}
  {\bf (Orientation of circles)}
  The orientation of the circles of an octahedral point-circle configuration is
  chosen in such a manner that the corresponding orientation of the 
  faces of the octahedron is the same for all faces (when viewed from `outside')
  (cf.\ figure \ref{notation}).
\end{convention}
\begin{figure}
\centerline{\includegraphics[width=\textwidth]{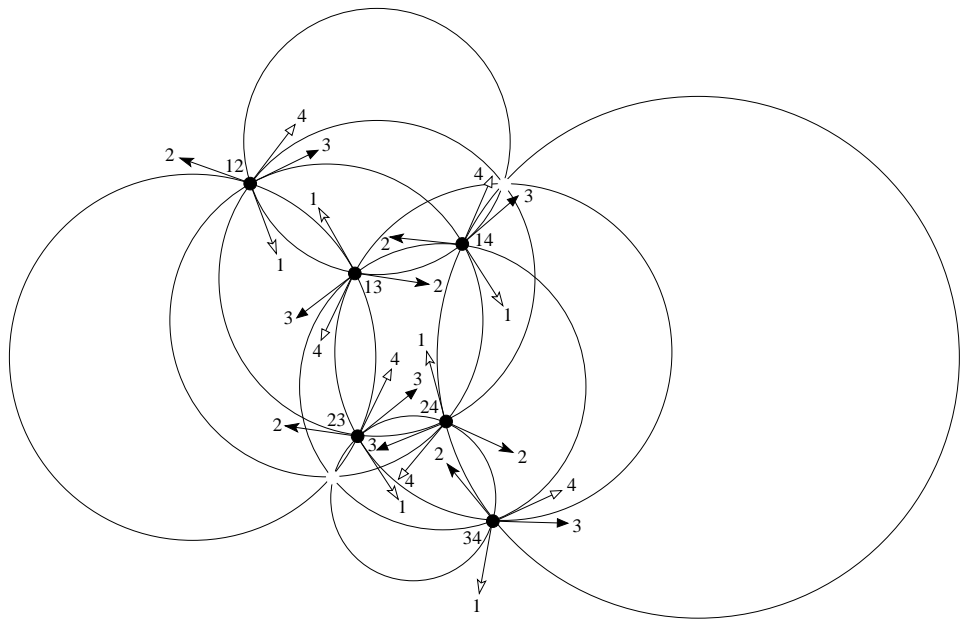}}
\caption{An admissible orientation of circles and tangent vectors}
\label{notation}
\end{figure}

\noindent
For completeness, it is remarked that the most general admissible orientation
of the circles is obtained by simultaneously changing the orientation of
a face of the octahedron and its three neighbours and then iterating this 
operation.

If we now demand that an oriented circle and its tangent vectors be of the same
orientation then the following definition is natural:

\begin{definition}
  {\bf (Angles between circles)}
  The angle made by two oriented circles $S$ and $S'$ passing through a point
  $P$ of an octahedral point-circle configuration is that made by the two
  corresponding tangent vectors $V$ and $V'$ at $P$, viz
\bela{E4}
   \angle(S,S') := \angle(V,V'). 
\ela
\end{definition}

We are now in a position to define an analogue of Clifford's classical
configuration:

\begin{definition}\label{cliff}
  {\bf (Generalized Clifford configurations)}
  An octahedral point-circle configuration is termed a {\em generalized Clifford
  configuration} if the six points $P_k$ are {\em equivalent} in the sense that
  for any six pairs of corresponding oriented circles $S_k,S'_k$ passing through
  $P_k$ the angle $\angle(S_k,S'_k)$ is independent of $k$.
\end{definition}

The following theorem demonstrates that the above definition is natural:

\begin{theorem} {\bf (Planar generalized Clifford configurations)}
  Generalized Clifford configurations on the plane coincide with classical
  Clifford configurations.
\end{theorem}

\begin{proof} 
The Godt-Ziegenbein property 
(Godt 1896; Ziegenbein 1941) consists of the equivalence of the points
$P,\ldots,P_{1234}$ of a classical Clifford configuration. Thus, in particular,
the points $P_{12},P_{13},P_{14},P_{23},P_{24},P_{34}$ of any
given Clifford configuration are equivalent and hence the above definition of 
a generalized Clifford configuration is met. Conversely, let 
$P_{12},P_{13},P_{14},P_{23},P_{24},P_{34}$ be the six points of a planar
generalized Clifford configuration. If we set set aside the point $P_{34}$,
say, then it is evident that, due to the assumption of equivalence, $P_{34}$ may
be reconstructed from the other five points 
$P_{12},P_{13},P_{14},P_{23},P_{24}$. Indeed, since the four circles passing 
through $P_{12}$ are determined by these five points, the angles made by all
pairs of circles are known so that, in turn, the remaining 
four circles are uniquely determined. On the other hand, the above five points
also belong to a unique classical Clifford configuration as discussed in the 
previous section. The latter has the Godt-Ziegenbein property and hence 
coincides with the given generalized Clifford configuration.
\end{proof}

\section{Quaternions}

It is well known that the group of conformal transformations in $\mathbbm{R}^n$,
$n>2$ consists of translations, rotations, scalings and inversions 
(Dubrovin \textit{et al.} 1984). 
Accordingly, generalized Clifford configurations are objects of conformal 
geometry since circles are mapped to circles. 
In the current context, it is therefore natural to identify the four-dimensional 
Euclidean space $\mathbbm{R}^4$ with the algebra of quaternions~$\mathbbm{H}$
(Koecher \& Remmert 1991). 
Thus, we adopt the quaternionic representation
\bela{E5}
   \mathbbm{R}^4\ni (a,b,c,d)\quad \leftrightarrow\quad
               (a\mathbbm{1} + b\,\mathbbm{i} + c\,\mathbbm{j}
               + d\,\mathbbm{k})\in\mathbbm{H},
\ela
where the matrices $\mathbbm{1},\mathbbm{i},\mathbbm{j},\mathbbm{k}$ are
defined by
\bela{E6}
   \mathbbm{1} = \left(\bear{cc}1&0\\ 0&1\ear\right),\quad
   \mathbbm{i} = \left(\bear{cc}0&-\i\\ -\i&0\ear\right),\quad
   \mathbbm{j} = \left(\bear{cc}0&-1\\ 1&0\ear\right),\quad
   \mathbbm{k} = \left(\bear{cc}-\i&0\\ 0&\i\ear\right).
\ela
Then, the following properties and identities may be established.

Firstly, it is readily verified that
\bela{E7}
   |X| = \sqrt{\det X},\quad XX^{\dagger} = \det X\,\mathbbm{1},\quad
   <X,Y> = \frac{1}{2}\tr(XY^{\dagger}). 
\ela
In particular, if, for non-vanishing quaterninons, we denote the corresponding
unit vectors by
\bela{E8} 
  \hat{X} = \frac{X}{|X|}
\ela
then
\bela{E9}
  \cos\angle(X,Y) = \frac{1}{2}\tr(\hat{X}\hat{Y}^{\dagger})
                  = \frac{1}{2}\tr(\hat{X}\hat{Y}^{-1}) = \frac{1}{2}
                  \frac{\tr(XY^{\dagger})}{\sqrt{\det X}\sqrt{\det Y}}.
\ela
Secondly, Cayley's theorem (Koecher \& Remmert 1991) states that any element 
$\Omega$ of the orthogonal group $O(4)$ is represented by either
\bela{E10}
   X \mapsto \hat{A}X\hat{B}\quad {\rm or} \quad
   X \mapsto \hat{A}X^{\dagger}\hat{B},\qquad 
  A,B\in\mathbbm{H}\backslash \{ 0\},
\ela
depending on whether $\Omega$ is `proper' ($\det\Omega=1$) or `improper' ($\det
\Omega=-1$) respectively. Conversely, any quaternionic action of the above type
corresponds to an orthogonal mapping $\Omega$. In particular, the operation
\bela{E11}
   X \mapsto \hat{A}X^{\dagger}\hat{A},\quad A\in\mathbbm{H}
\ela
constitutes a reflection in the vector $A$. Finally, the group of conformal 
transformations is generated by the orientation-preserving
M\"obius transformations (Ahlfors 1981)
\bela{E12}
  \mathcal{M} : X \mapsto (AX + B)(CX+D)^{-1},\qquad A,B,C,D \in\mathbbm{H}
\ela
and the particular reflection (conjugation)
\bela{E13}
  \mathcal{C} : X \mapsto X^{\dagger}.
\ela

\section{Quaternionic cross- and multi-ratios}

The relevance of multi-ratios in connection with classical Clifford 
configurations has been indicated in \S 2. It turns out that generalized 
Clifford configurations may also be described algebraically in terms of 
quaternionic multi-ratios. It is recalled that the cross-ratio of four points 
in $\mathbbm{R}^4$ is usually taken to be (Ahlfors 1981)
\bela{E14}
  Q(P_1,P_2,P_3,P_4) = (P_1-P_2)(P_2-P_3)^{-1}(P_3-P_4)(P_4-P_1)^{-1}.
\ela
The cross-ratio is `real', that is
\bela{E15}
  Q(P_1,P_2,P_3,P_4) = a\mathbbm{1},\quad a\in \mathbbm{R},
\ela
if and only if the four points $P_1,\ldots,P_4$ lie on a circle.
The multi-ratio of six points may be defined as
\bela{E16}
  \bear{l}
   M(P_1,P_2,P_3,P_4,P_5,P_6)\as
   \quad=(P_1-P_2)(P_2-P_3)^{-1}(P_3-P_4)(P_4-P_5)^{-1}(P_5-P_6)(P_6-P_1)^{-1}
  \ear
\ela
or, alternatively,
\bela{E17} 
 \bear{l}
   \tilde{M}(P_1,P_2,P_3,P_4,P_5,P_6)\as
   \quad=(P_1-P_6)^{-1}(P_6-P_5)(P_5-P_4)^{-1}(P_4-P_3)(P_3-P_2)^{-1}(P_2-P_1).
 \ear
\ela
In the scalar case, the `right-multi-ratio' $M$ and the `left-multi-ratio' 
$\tilde{M}$ are evidently identical. However, in the 
quaternionic case, the two multi-ratios are related by 
\bela{E18}
   \tilde{M}(P_1,P_2,P_3,P_4,P_5,P_6) = (P_1-P_6)^{-1}
           M(P_6,P_5,P_4,P_3,P_2,P_1)(P_1-P_6).
\ela
This relation shows that the multi-ratio {\em conditions}
\bela{E19}
  M(P_1,P_2,P_3,P_4,P_5,P_6) = -\mathbbm{1}
\ela
and
\bela{E20}
  \tilde{M}(P_6,P_5,P_4,P_3,P_2,P_1) = -\mathbbm{1}
\ela
on six points $P_1,\ldots,P_6$ coincide. Furthermore, it is readily seen that
the M\"obius transformations $\mathcal{M}$ individually 
preserve the multi-ratio conditions (\ref{E19})
and
\bela{E21}
   \tilde{M}(P_1,P_2,P_3,P_4,P_5,P_6) = -\mathbbm{1},
\ela
while (\ref{E19}) and (\ref{E21}) are 
mapped to each other by the conformal transformations
$\mathcal{C}\circ\mathcal{M}$ which change the orientation. The geometric
significance of this fact in the context of generalized Clifford configurations
is discussed in the following section.

\section{Existence and algebraic description of generalized Clifford
configurations}

It has been shown in \S 2 that planar generalized Clifford configurations
are uniquely determined by five points. It turns out that the derivation
of an analogous statement in the general case is the key to an algebraic 
description of generalized Clifford configurations. To this end,
it is convenient to make a canonical choice of the tangent vectors to oriented
circles. Thus, since four points $X,A,B,C$ are concyclic if and only if their
cross-ratio is real, the function $X(s)$ defined by
\bela{E22}
  Q(X,A,B,C) = s\mathbbm{1},\quad s\in\mathbbm{R}
\ela
parametrizes the oriented circle $S_{A,B,C}$ which passes through $A,B,C$ as
$s$ increases with $X(0)=A,\,X(1)=B,\,X(\infty)=C$. Differentiation and
evaluation at $s=1$ then results in the tangent vector
\bela{E23}
  V_{A,B,C} = (C-B)(C-A)^{-1}(B-A) = (B-A)(C-A)^{-1}(C-B)
\ela
at the point $X=B$. The latter identity is merely a property of any three
matrices $A,B$ and $C$. Moreover, if two oriented circles $S_1$ and $S_2$ meet
at the points $P$ and $P'$ with associated tangent vectors $V_1$ and $V_2$ at
$P$ then the vectors $V_1'$ and $V_2'$ given by (cf.\ (\ref{E11}))
\bela{E24}
   V_1' = (P'-P)V_1^{-1}(P'-P),\quad
   V_2' = (P'-P)V_2^{-1}(P'-P)
\ela
are tangent to the circles $S_1$ and $S_2$ at $P'$ and the orientation of the
tangent vectors is preserved.

In order to proceed, we introduce a natural labelling of the points of an
octahedral point-circle configuration and the vertices of its underlying
octahedron. Thus, the six vertices of the octahedron are labelled by
$(ik)=(ki)$, $i\neq k\in \{1,2,3,4\}$  in such a way that opposite vertices
carry complementary indices and the corresponding points of the configuration
are denoted by $\Phi_{ik}=\Phi_{ki}$ throughout the remainder of this paper.

\begin{theorem}\label{lemma1} {\bf (`Uniqueness')}
  There exist at most two generalized Clifford configurations which share five
  points and the four associated circles.
\end{theorem}

\begin{proof}
For convenience, we label the five common points by
$\Phi_{12},\Phi_{13},\Phi_{14},\Phi_{23},\Phi_{24}$ and regard the
point(s) $\Phi_{34}$ as unknown. Accordingly, only the four circles 
$S^{\mu}_{12}$, $\mu=1,2,3,4$ passing through the point $\Phi_{12}$ are known. 
The tangent vectors to these circles are linearly dependent and may be chosen to 
be
\bela{E25}
   V_{12}^1 = V_{14,12,13},\quad V_{12}^2 = V_{13,12,23},\quad
   V_{12}^3 = V_{24,12,14},\quad V_{12}^4 = V_{23,12,24}
\ela
as indicated in figure \ref{notation}. Here, the notation $V_{14,12,13}=
V_{\Phi_{14},\Phi_{12},\Phi_{13}}$ etc.\ has been adopted. In the generic case, 
any three of these four vectors span the three-dimen\-sional tangent hyperplane 
to the hypersphere at $\Phi_{12}$
defined by the five points \linebreak
$\Phi_{12},\Phi_{13},\Phi_{14},\Phi_{23},\Phi_{24}$. 
By virtue of the correspondence principle and the 
orientation convention, the orientation of the eight circles is now defined. 

The vectors
\bela{E26}
 \bear{c}
   V^1_{13} = (\Phi_{13}-\Phi_{12})(V^2_{12})^{-1}(\Phi_{13}-\Phi_{12})\as
   V^2_{13} = (\Phi_{13}-\Phi_{12})(V^1_{12})^{-1}(\Phi_{13}-\Phi_{12})
   \ear
\ela
are tangent to the circles $S^1_{13}=S^2_{12}$ and $S^2_{13}=S^1_{12}$
respectively at the point $\Phi_{13}$, while
\bela{E27}
   V^1_{14} = (\Phi_{14}-\Phi_{12})(V^3_{12})^{-1}(\Phi_{14}-\Phi_{12})
\ela
constitutes a tangent vector to the circle $S^1_{14}=S^3_{12}$ at the point
$\Phi_{14}$. Now, in order to determine the two remaining tangent vectors at the
point $\Phi_{13}$, we make use of the assumption that, in particular, the points 
$\Phi_{12},\Phi_{13}$ and $\Phi_{14}$ are equivalent. Thus, the circle 
$S^3_{13}$ which passes through the points
$\Phi_{13}$ and $\Phi_{14}$ gives rise to the relations
\bela{A1}
  \angle(S^3_{13},S^1_{13}) = \angle(S^3_{12},S^1_{12}),\quad
   \angle(S^3_{13},S^2_{13}) = \angle(S^3_{12},S^2_{12}),
\ela
while the circle $S^2_{14}=S^3_{13}$ is associated with the additional
relation
\bela{A2}
  \angle(S^2_{14},S^1_{14}) = \angle(S^2_{12},S^1_{12}),
\ela
where $S^1_{14}=S^3_{12}$. If the tangent vector to $S^3_{13}$ at the point
$\Phi_{13}$ is denoted by~$V^3_{13}$ then the vector
\bela{E28}
   V^2_{14} = (\Phi_{14}-\Phi_{13})(V^3_{13})^{-1}(\Phi_{14}-\Phi_{13})
\ela
is tangent to the circle $S^2_{14}$ at the point $\Phi_{14}$. Accordingly, the
conditions (\ref{A1}) and (\ref{A2}) translate into
\bela{E29}
  \bear{l}
    \tr[\hat{V}^3_{13}(\hat{V}^1_{13})^{\dagger}] =
    \tr[\hat{V}^3_{12}(\hat{V}^1_{12})^{\dagger}]\as
    \tr[\hat{V}^3_{13}(\hat{V}^2_{13})^{\dagger}] =
    \tr[\hat{V}^3_{12}(\hat{V}^2_{12})^{\dagger}]\as
    \tr[\hat{V}^2_{14}(\hat{V}^1_{14})^{\dagger}] =
    \tr[\hat{V}^2_{12}(\hat{V}^1_{12})^{\dagger}]
   \ear
\ela
which may be written as
\bela{L1}
  \bear{rl}
    \tr[\hat{\Delta}(\hat{V}^2_{12})^{-1}] &=
    \tr[\hat{V}^1_{12}(\hat{V}^3_{12})^{-1}]\as
    \tr[\hat{\Delta}(\hat{V}^1_{12})^{-1}] &=
    \tr[\hat{V}^2_{12}(\hat{V}^3_{12})^{-1}]\as
    \tr[\hat{\Delta}(\hat{V}^1_{12})^{-1}\hat{V}^3_{12}(\hat{V}^1_{12})^{-1}] &=
    \tr[\hat{V}^2_{12}(\hat{V}^1_{12})^{-1}]
   \ear
\ela
with the definition
\bela{D1}
  \Delta = (\Phi_{13}-\Phi_{12})(V^3_{13})^{-1}(\Phi_{13}-\Phi_{12}).
\ela
The constraints (\ref{L1}) constitute a linear system for the unit vector
$\hat{\Delta}$. Hence, there are two cases to distinguish:
\medskip

\noindent
{\bf Case 1:} If the tangent vectors $V^{\mu}_{12}$ span a two-dimensional plane
then the five points 
$\Phi_{12},\Phi_{13},\Phi_{14},\Phi_{23},\Phi_{24}$ necessarily lie on a
2-sphere (or a plane). On use of a conformal transformation, this 
2-sphere may be mapped to a plane so that we are left with the consideration
of generalized Clifford configurations in a three-dimensional Euclidean space 
subject to
the five points $\Phi_{12},\Phi_{13},\Phi_{14},\Phi_{23},\Phi_{24}$ being 
co-planar. If there exists a generalized Clifford configuration for which
$\Phi_{34}$ does not lie on the corresponding plane $\Sigma$ then reflection of
$\Phi_{34}$ in $\Sigma$ produces another generalized Clifford 
configuration.
Furthermore, the classical Clifford configuration defined uniquely by the
five points constitutes a third (planar) generalized Clifford configuration.
Thus, the number of distinct generalized Clifford configurations 
sharing five co-planar points is odd.

On the other hand, it is readily seen that the rank of 
the linear system (\ref{L1}) is 2 since the tangent vectors
$V^1_{12},V^2_{12},V^3_{12}$ are linearly dependent and hence there exist at 
most two solutions $\hat{\Delta}$.
Any specific choice of $\hat{\Delta}$ determines the tangent vector
$V^3_{13}$ up to its magnitude. Moreover, since the angles between the vectors
$V^1_{13},V^2_{13},V^3_{13}$ and the vector $V^4_{13}$ which is tangent to the
circle $S^4_{13}$ passing through the points $\Phi_{23}$ and $\Phi_{13}$ are
known and the four vectors $V^{\mu}_{13}$ must be linearly dependent, the
direction of the tangent vector $V^4_{13}$ is uniquely determined. This, in 
turn, shows that the point $\Phi_{34}$ is unique. Thus, there exist at most two
generalized Clifford configurations. We therefore conclude that the 
above-mentioned classical Clifford configuration is the only
generalized Clifford configuration under the current assumption.
\medskip

\noindent
{\bf Case 2:} If the tangent vectors $V^{\mu}_{12}$ span a three-dimensional 
vector space then the tangent vectors $V^1_{12},V^2_{12},V^3_{12}$ are linearly
independent without loss of generality. Accordingly, the rank of the linear 
system (\ref{L1}) is 3 and the corresponding two solutions are given by
\bela{D2}
  \hat{\Delta}_1 = \hat{V}^1_{12}(\hat{V}^3_{12})^{-1}\hat{V}^2_{12}\quad
  {\rm and}\quad
  \hat{\Delta}_2 = \hat{V}^2_{12}(\hat{V}^3_{12})^{-1}\hat{V}^1_{12}.
\ela
It is noted that the above solutions are distinct since equality implies 
that \linebreak
$[\hat{V}^1_{12}(\hat{V}^3_{12})^{-1},\hat{V}^2_{12}(\hat{V}^3_{12})^{-1}]=0$
and hence $V^1_{12},V^2_{12},V^3_{12}$ are linearly dependent. In fact, since
the projections of $\hat{\Delta}_1$ and $\hat{\Delta}_2$ onto 
$V^1_{12}, V^2_{12}$ and $V^3_{12}$ coincide, the unit vectors 
$\hat{\Delta}_1$ and $\hat{\Delta}_2$ are related by reflection in the 
three-dimensional vector space spanned by $V^1_{12},V^2_{12},V^3_{12}$. As in 
the previous case, any specific choice of $\hat{\Delta}$ now determines the 
point $\Phi_{34}$ uniquely so that there exist at most two generalized Clifford
configurations.
\end{proof}

Remarkably, the above analysis implies that generalized Clifford configurations
in three-dimensional Euclidean spaces or spheres are essentially 
two-dimensional.

\begin{theorem}
  Any generalized Clifford configuration in a three-dimensional Euclidean space
  $\mathbbm{R}^3$ or a three-dimensional sphere is either planar or 
  confined to a two-dimensional sphere and may therefore be
  mapped to a classical Clifford configuration by means of a suitable conformal
  transformation.
\end{theorem}

\begin{proof}
Since any generalized Clifford configuration in a 
three-dimensional hypersphere may be mapped via a conformal transformation to a 
generalized Clifford configuration in a three-dimensional subspace of 
$\mathbbm{R}^4$, we may confine ourselves to the
case of a generalized Clifford configuration in 
$\mathbbm{R}^3$ which we denote by $E$ and regard 
as being embedded in $\mathbbm{R}^4$. As in the preceding proof, we consider the
five points $\Phi_{12},\Phi_{13},\Phi_{14},\Phi_{23},\Phi_{24}$ and set aside 
the point $\Phi_{34}$. If the tangent vectors $V^{\mu}_{12}$ span a 
two-dimensional plane then the generalized Clifford configuration is planar or
confined to a 2-sphere as shown above. Hence, we may assume without loss of
generality that $E=\mbox{span}\{V^1_{12},V^2_{12},V^3_{12}\}$. This implies, 
in turn, that the two solutions
$\hat{\Delta}_1$ and $\hat{\Delta}_2$ of the linear system (\ref{L1}) are 
distinct and are related by reflection in $E$. Hence, 
$\hat{\Delta}_1,\hat{\Delta}_2\not\in E$.
\end{proof}

The key observation is now the following:

\begin{theorem}\label{main}
  {\bf (Multi-ratio description of generalized Clifford configurations)}
  The points $\Phi_{ik}$ of a generalized Clifford configuration are related by
  either
  \bela{M1}
    M(\Phi_{14},\Phi_{12},\Phi_{24},\Phi_{23},\Phi_{34},\Phi_{13}) =
    -\mathbbm{1}
  \ela
  or
  \bela{M2}
    \tilde{M}(\Phi_{14},\Phi_{12},\Phi_{24},\Phi_{23},\Phi_{34},\Phi_{13}) =
    -\mathbbm{1}.
  \ela
  Conversely, any six points $\Phi_{ik}$ of an octahedral point-circle
  configuration which obey either of the multi-ratio
  conditions (\ref{M1}) or (\ref{M2}) constitute the points of a
  generalized Clifford configuration.
\end{theorem}

Before we prove the theorem, it is observed
that the above multi-ratio conditions are representatives of two equivalence
classes of multi-ratio conditions which may be imposed on an octahedral
point-circle configuration. Specifically, let us consider an oriented hexagon 
$(\alpha_1,\alpha_2,\alpha_3,\alpha_4,\alpha_5,\alpha_6)$ with fixed
`initial' vertex $\alpha_1$ formed by six 
edges of an octahedron with distinct vertices $\alpha_i$ in such a way that any 
two adjacent edges of the hexagon belong to a triangular face of the
octahedron (cf.\ figure \ref{hexagon}) and impose the multi-ratio condition
\bela{E30}
  M(P_{\alpha_1},P_{\alpha_2},P_{\alpha_3},P_{\alpha_4},P_{\alpha_5},
    P_{\alpha_6}) = -\mathbbm{1}
\ela
on the points $P_{\alpha_i}$ of an associated octahedral point-circle 
configuration. There exist $48$ such hexagons and in the planar case the
associated multi-ratio conditions are equivalent. However, the situation is
different in the (generic) quaternionic case.
\begin{figure}
\centerline{\includegraphics[width=0.8\textwidth]{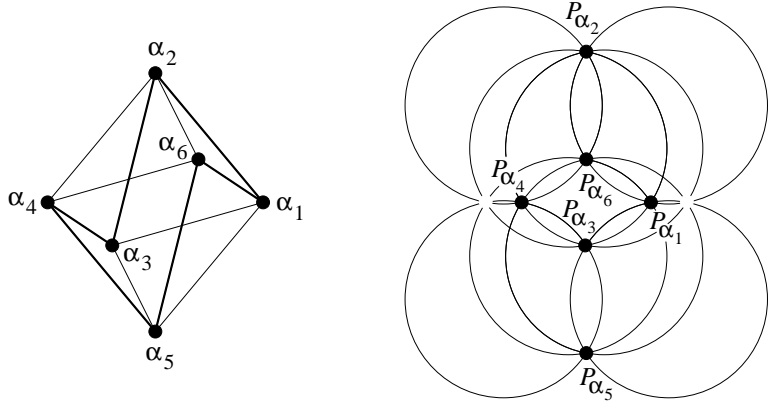}}
\caption{An oriented hexagon 
$(\alpha_1,\alpha_2,\alpha_3,\alpha_4,\alpha_5,\alpha_6)$ associated with the
multi-ratio condition $M(P_{\alpha_1},P_{\alpha_2},P_{\alpha_3},P_{\alpha_4},
P_{\alpha_5},P_{\alpha_6}) = -\mathbbm{1}$}
\label{hexagon}
\end{figure}

\begin{lemma}\label{perm}
  The multi-ratio conditions (\ref{M1}) and (\ref{M2}) are either invariant or
  mapped to each other by the associated octahedral symmetry group. 
  In particular, the subgroup of symmetries which leave the multi-ratio
  conditions invariant consists of the permutations of the indices 
  $1,2,3,4$.
\end{lemma}

\begin{proof}
On the one hand, a typical permutation which leads to a 
non-trivial ordering of the arguments in the multi-ratio conditions is given by
$(1,2,3,4)\rightarrow(4,2,3,1)$. The associated multi-ratio condition
\bela{E30a}
  M(\Phi_{14},\Phi_{24},\Phi_{12},\Phi_{23},\Phi_{13},\Phi_{34}) = -\mathbbm{1}
\ela
may be brought into the form
\bela{E30b}
  (\Phi_{14}-\Phi_{24})(\Phi_{24}-\Phi_{12})^{-1}(\Phi_{12}-\Phi_{23})
 + (\Phi_{34}-\Phi_{14})(\Phi_{13}-\Phi_{34})^{-1}(\Phi_{23}-\Phi_{13})
 = 0.
\ela
The identities
\bela{E30c}
 \bear{c}
   (\Phi_{14}-\Phi_{24})(\Phi_{24}-\Phi_{12})^{-1}(\Phi_{12}-\Phi_{23})
   \qquad\qquad\qquad\qquad\as 
   \qquad\qquad
   =(\Phi_{14}-\Phi_{12})(\Phi_{24}-\Phi_{12})^{-1}(\Phi_{24}-\Phi_{23})
   + \Phi_{23} - \Phi_{14}\as
   (\Phi_{34}-\Phi_{14})(\Phi_{13}-\Phi_{34})^{-1}(\Phi_{23}-\Phi_{13})
   \qquad\qquad\qquad\qquad\as 
   \qquad\qquad
   =(\Phi_{13}-\Phi_{14})(\Phi_{13}-\Phi_{34})^{-1}(\Phi_{23}-\Phi_{34})
   + \Phi_{14} - \Phi_{23}
 \ear
\ela
then lead to 
\bela{E30d}
  (\Phi_{14}-\Phi_{12})(\Phi_{24}-\Phi_{12})^{-1}(\Phi_{24}-\Phi_{23})
 + (\Phi_{13}-\Phi_{14})(\Phi_{13}-\Phi_{34})^{-1}(\Phi_{23}-\Phi_{34})
 = 0
\ela
which, in turn, is the original multi-ratio condition (\ref{M1}).

On the other hand, composition of the subgroup of permutations with any
`rotation by 90 degrees' of the octahedron generates the remaining 24
discrete symmetries. For instance, the rotation 
$({14},{12},{24},{23},{34},{13})
\rightarrow 
({24},{12},{23},{13},{34},{14})$
produces the multi-ratio condition
\bela{E30e}
  M(\Phi_{24},\Phi_{12},\Phi_{23},\Phi_{13},\Phi_{34},\Phi_{14}) = -\mathbbm{1}
\ela
which is, by definition, equivalent to 
\bela{E30f}
  \tilde{M}(\Phi_{14},\Phi_{24},\Phi_{12},\Phi_{23},\Phi_{13},\Phi_{34}) 
  = -\mathbbm{1}.
\ela
The latter is merely the tilde version of
(\ref{E30a}) and hence equivalent to the multi-ratio condition 
(\ref{M2}).
\end{proof}

We are now in a position to prove theorem \ref{main}.

\begin{proof} {\sl (of theorem \ref{main})} 
By virtue of theorem \ref{lemma1} and the fact that either
of the multi-ratio conditions (\ref{M1}) or (\ref{M2}) may be regarded as a
definition of a point in terms of five arbitrary points, it remains to show
that the multi-ratio conditions indeed give rise to generalized Clifford
configurations. Thus, we here assume that the six points $\Phi_{ik}$ of an
octahedral point-circle configuration are constrained by one of the multi-ratio
conditions and choose the tangent vectors $V^{\mu}_{12}$ as in the proof of
lemma \ref{lemma1}. Accordingly, a set of corresponding tangent vectors
$V^{\mu}_{13}$ of the correct orientation at the point $\Phi_{13}$ is given by
(cf.\ figure \ref{notation})
\bela{E31}
   V_{13}^1 = V_{23,13,12},\quad V_{13}^2 = V_{12,13,14},\quad
   V_{13}^3 = V_{14,13,34},\quad V_{13}^4 = V_{34,13,23}.
\ela
The expressions (\ref{D1}) and (\ref{D2}) for the quantities $\Delta$ and
$\hat{\Delta}_1,\hat{\Delta}_2$ suggest that one should consider the
orientation- and angle-preserving mappings
\bela{E32}
  \bear{l}
    \mathcal{O}_1 : X \mapsto (\Phi_{13}-\Phi_{12})(V^2_{12})^{-1}X
                             (V^1_{12})^{-1}(\Phi_{13}-\Phi_{12})\as

    \mathcal{O}_2 : X \mapsto (\Phi_{13}-\Phi_{12})(V^1_{12})^{-1}X
                             (V^2_{12})^{-1}(\Phi_{13}-\Phi_{12}).
   \ear
\ela
Indeed, it is readily verified that
\bela{E33}
   V^{\nu}_{13} = \mathcal{O}_i(V^{\nu}_{12}),\quad \nu=1,2 
\ela
(cf.\ (\ref{E26})) and a short calculation reveals that
\bela{E34}
   \bear{rl}
     \mathcal{O}_1(V^{3}_{12})(V^{3}_{13})^{-1} & =
     - M(\Phi_{13},\Phi_{23},\Phi_{12},\Phi_{24},\Phi_{14},\Phi_{34})\as
     (V^{4}_{13})^{-1}\mathcal{O}_1(V^{4}_{12}) & =
     - \tilde{M}(\Phi_{13},\Phi_{14},\Phi_{12},\Phi_{24},\Phi_{23},\Phi_{34})\as
     (V^{3}_{13})^{-1}\mathcal{O}_2(V^{3}_{12}) & =
     - \tilde{M}(\Phi_{13},\Phi_{23},\Phi_{12},\Phi_{24},\Phi_{14},\Phi_{34})\as
     \mathcal{O}_2(V^{4}_{12})(V^{4}_{13})^{-1} & =
     - M(\Phi_{13},\Phi_{14},\Phi_{12},\Phi_{24},\Phi_{23},\Phi_{34}).

   \ear
\ela

Two conclusions may now be drawn from lemma \ref{perm}. Firstly, 
the conditions (\ref{M1}) and (\ref{M2}) are equivalent to
\bela{E35}
 \bear{rl}
  M(\Phi_{13},\Phi_{23},\Phi_{12},\Phi_{24},\Phi_{14},\Phi_{34})& 
  = -\mathbbm{1}\as
  \Leftrightarrow\quad
  \tilde{M}(\Phi_{13},\Phi_{14},\Phi_{12},\Phi_{24},\Phi_{23},\Phi_{34}) & 
  = -\mathbbm{1}
 \ear
\ela
and
\bela{E36}
 \bear{rl}
  \tilde{M}(\Phi_{13},\Phi_{23},\Phi_{12},\Phi_{24},\Phi_{14},\Phi_{34}) &
  = -\mathbbm{1}\as
  \Leftrightarrow\quad
  M(\Phi_{13},\Phi_{14},\Phi_{12},\Phi_{24},\Phi_{23},\Phi_{34}) & 
  = -\mathbbm{1}
 \ear
\ela
respectively and hence it follows that
\bela{E37}
   V^{\mu}_{13} = \mathcal{O}_1(V^{\mu}_{12})\quad{\rm or}\quad
   V^{\mu}_{13} = \mathcal{O}_2(V^{\mu}_{12}),\quad \mu=1,2,3,4
\ela
depending on whether (\ref{M1}) or (\ref{M2}) is assumed to hold. Thus, the
points $\Phi_{12}$ and $\Phi_{13}$ are equivalent in the sense of definition
\ref{cliff} since
\bela{E38}
   \angle(V^{\mu}_{12},V^{\mu'}_{12}) = \angle(V^{\mu}_{13},V^{\mu'}_{13}),
\quad \mu,\mu'\in\{1,2,3,4\}.
\ela
Secondly, since the points $\Phi_{12}$ and $\Phi_{13}$ are equivalent, the
points $\Phi_{\pi(1)\pi(2)}$ and $\Phi_{\pi(1)\pi(3)}$ are equivalent for any
permutation $\pi$ of the indices $1,2,3,4$. Hence, all points are equivalent. 
This completes the proof.
\end{proof}

We conclude this section with the remark that if a point of a generalized
Clifford configuration which is not conformally equivalent to
a classical Clifford configuration is inverted with respect to the hypersphere 
which passes through the other five points then another generalized Clifford 
configuration is obtained. Thus, theorem~\ref{main} implies the following 
corollary:

\begin{corollary}
If two generalized Clifford configurations defined by
\bela{E39}
 \bear{l}
      M(\Phi_{14},\Phi_{12},\Phi_{24},\Phi_{23},\Phi_{34},\Phi_{13}) =
      -\mathbbm{1}\as
      \tilde{M}(\Phi_{14},\Phi_{12},\Phi_{24},\Phi_{23},
                          \tilde{\Phi}_{34},\Phi_{13}) = -\mathbbm{1}
     \ear
\ela
are not conformally equivalent to a classical Clifford configuration then the
points $\Phi_{34}$ and $\tilde{\Phi}_{34}$ 
are related by inversion with respect to the (well-defined)
hypersphere passing through
the common points $\Phi_{13},\Phi_{14},\Phi_{12},\Phi_{24},\Phi_{23}$.
\end{corollary}

Interestingly, if one chooses the points 
$\Phi_{13},\Phi_{14},\Phi_{12},\Phi_{24},\Phi_{23}$ in such a way that 
$\tilde{\Phi}_{34}$ lies `at infinity' then the above corollary implies that 
the point $\Phi_{34}$ constitutes the centre of the hypersphere which passes 
through these five points. 

\section{The quaternionic discrete Schwarzian KP equation}

In Konopelchenko \& Schief (2002), it has been demonstrated that the multi-ratio
condition (\ref{E1}) interpreted as a lattice equation is nothing but a
Schwarzian version of the discrete Kadomtsev-Petviashvili (dKP) equation which
constitutes a `master equation' in the theory of integrable systems 
(Ablowitz \& Segur 1981; Zakharov \textit{et al.} 1980). 
Accordingly, the dSKP equation admits a geometric interpretation in terms of
classical Clifford configurations. Here, we present an analogous construction
of three-dimensional `Clifford lattices' in a four-dimensional Euclidean space
associated with the quaternionic multi-ratio condition. 
Due to the absence of any
additional points of intersection of the circles belonging to a generic 
generalized Clifford configuration, not only is the quaternionic case more 
natural but it also includes the afore-mentioned scalar case.   

We consider lattices of the combinatorics of a face-centred cubic (fcc) lattice
in a four-dimensional Euclidean space, that is maps of the form
\bela{E40}
  \Phi : \mathbbm{F}\rightarrow \mathbbm{H},\quad
  \mathbbm{F} = \{(n_1,n_2,n_3)\in\mathbbm{Z}^3: n_1 + n_2 + n_3 \mbox{ odd}\}.
\ela   
\begin{figure}
\centerline{\includegraphics[width=0.45\textwidth]{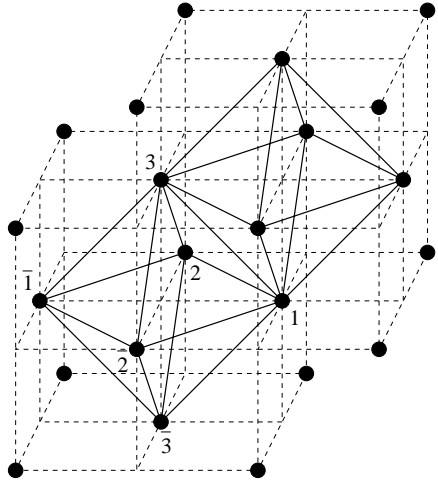}}
\caption{Two octahedra embedded in the fcc lattice $\mathbbm{F}$}
\label{cube}
\end{figure}
Any six points of $\mathbbm{F}$ which constitute the centres of the faces of
a cube composed of eight adjacent elementary cubes of $\mathbbm{Z}^3$ 
may be regarded as the vertices of an octahedron (cf.\ figure \ref{cube}). 
In this way, the fcc lattice may be identified with the vertices of a 
collection of octahedra which meet at common edges. Accordingly, 
$\Phi(\mathbbm{F})$ constitutes a set of points belonging to octahedral
point-circle configurations. It is therefore natural to require
that these point-circle configurations be of generalized Clifford type. In 
terms of the labelling (\ref{E40}), this implies that for any octahedron 
centred at $(\nu_1,\nu_2,\nu_3)\in\mathbbm{Z}^3$ with 
$\nu_1+\nu_2+\nu_3$ even, one of the multi-ratio conditions
\bela{E41}
  M(\Phi_{\bar{1}},\Phi_{2},\Phi_{\bar{3}},\Phi_{1},\Phi_{\bar{2}},\Phi_{3})
  = -\mathbbm{1}
\ela
and
\bela{E42}
  \tilde{M}
  (\Phi_{\bar{1}},\Phi_{2},\Phi_{\bar{3}},\Phi_{1},\Phi_{\bar{2}},\Phi_{3})
  = -\mathbbm{1}
\ela
obtains, where the arguments of $\Phi$ have been suppressed and the notation
\bela{E43}
  \Phi = \Phi(\nu_1,\nu_2,\nu_3),\quad
  \Phi_{\bar{1}} =  \Phi(\nu_1-1,\nu_2,\nu_3),\quad
  \Phi_{2} = \Phi(\nu_1,\nu_2+1,\nu_3),\quad\ldots
\ela
has been adopted. Two well-defined maps are obtained by demanding that the
`same' multi-ratio condition is imposed on all octahedra. Thus, we say that the
map $\Phi$ defines a {\em Clifford lattice} in $\mathbbm{H}$ if either
(\ref{E41}) or (\ref{E42}) regarded as a lattice equation holds. It is noted 
that these two equations are essentially identical in the sense that if 
$\Phi(n_1,n_2,n_3)$ is a solution of (\ref{E41}) then 
$\Phi(-n_1,-n_2,-n_3)$ is a solution of (\ref{E42}) and vice versa.
  
If the fcc lattice $\mathbbm{F}$ is mapped to a simple cubic lattice 
$\mathbbm{Z}^3$ via the relabelling 
\bela{E44}
  \mathbbm{F}\ni (n_1,n_2,n_3)\quad\leftrightarrow\quad
                 (m_1,m_2,m_3)\in\mathbbm{Z}^3
\ela
defined by
\bela{E45}
  n_1 = m_2 + m_3 - 1,\quad n_2 = m_1 + m_3 - 1,\quad n_3 = m_1 + m_2 - 1
\ela
then the lattice equation (\ref{E41}) assumes the standard form of the 
quaternionic reduction 
\bela{E46}
  M(\Phi_\va{1},\Phi_{\va{1}\va{3}},\Phi_\va{3},\Phi_{\va{2}\va{3}},
  \Phi_\va{2},\Phi_{\va{1}\va{2}}) = -\mathbbm{1}
\ela
of the integrable multi-component dSKP equation 
(Bogdanov \& Konopelchenko 1998), wherein the subscripts on $\Phi$ now 
refer to unit increments of the variables $m_1,m_2,m_3$.

The conformal geometry of the qdSKP equation (\ref{E46}) and the associated
(discrete) Schwarzian Davey-Stewartson II hierarchy has been discussed
in detail in Konopelchenko \& Schief (2005). Therein, it has been shown 
that any standard discrete isothermic surface 
(Bobenko \& Pinkall 1996) may be extended via
a translational symmetry to a three-dimensional lattice in such a manner that
a (degenerate) Clifford lattice is obtained. In particular, Clifford lattices 
encapsulate discrete surfaces of constant mean curvature and discrete minimal
surfaces. Thus, apart from its significance in connection with (generalized)
Clifford configurations, the qdSKP equation also plays an important role in the
area of integrable discrete differential geometry (Bobenko \& Seiler 1999). 

The geometric integrability of the qdSKP equation (\ref{E46}) may be shown 
by embedding generalized Clifford configurations in four-dimensional lattices
of suitable combinatorics so that any three-dimensional `slice' constitutes
a generalized Clifford lattice. In the case of classical 
Clifford configurations, Desargues' classical theorem (Pedoe 1970) turns out to
be the key to the construction of a well-posed Cauchy problem for planar
Clifford lattices. However, a reformulation
of Desargues' theorem in terms of angles is required in order to formulate
a well-posed Cauchy problem for four-dimensional Clifford lattices in  
$\mathbbm{H}$. In this manner, the standard Lax representation 
(Bogdanov \& Konopelchenko 1998) and a B\"acklund 
transformation for the qdSKP equation may be derived geometrically.
A detailed discussion of this topic goes beyond the scope of
this paper and is consigned to a separate publication.

\begin{thedemobiblio}\smallskip{}
\item Ablowitz, M. J. \& Segur, H. 1981 \textit{Solitons and the 
  inverse scattering transform}, Philadelphia: SIAM.
\item Ahlfors, L. V. 1981 \textit{M\"obius transformations in several 
  dimensions}, Ordway Professorship Lectures in Mathematics, Univ.\ 
  Minnesota.
\item Bobenko, A. I. \& Pinkall, U. 1996 Discrete
  isothermic surfaces. \textit{J. reine angew. Math.} \textbf{475}, 187--208.
\item Bobenko, A. I. \& Seiler, R. (eds) 1999 \textit{Discrete
  integrable geometry and physics}, Oxford: Clarendon Press. 
\item Bogdanov, L. V. \& Konopelchenko, B. G. 1998 
  Analytic-bilinear approach to integrable hierarchies. II.\ Multicomponent KP 
  and 2D Toda lattice hierarchies. \textit{J. Math. Phys.} \textbf{39},
  4701--4728. 
\item Brannan, D. A., Esplen, M. F. \& Gray, J. J. 1999 \textit{Geometry}, 
  Cambridge University Press.
\item Brown, L. M. 1954 A configuration of points and spheres in
  four-dimensional space. \textit{Proc. R. Soc. Edinb.} A \textbf{64}, 145--149.
\item Clifford, W. K. 1871 A synthetic proof of Miquel's theorem.
  \textit{Oxford, Cambridge and Dublin Messenger of Math.} \textbf{5}, 
  124--141.
\item Cox, H. 1891 Application of Grassmann's Ausdehnungslehre to 
  properties of circles. \textit{Quart. J. Pure Appl. Math.} \textbf{25}, 1--71.
\item Coxeter, H. S. M. 1956 Review of Brown (1954). \textit{Math. Rev.} 
  \textbf{17}, 886.
\item de Longechamps, G. 1877 Note de g\'eometrie. \textit{Nouvelle Corresp.  
  Math\`emat.} \textbf{3}, 306--312; 340--347.
\item Dorfman, I. Ya \& Nijhoff, F. W. 1991 On a 
  2+1-dimensional version of the Krichever-Novikov equation. \textit{Phys. 
  Lett.} \textbf{157A}, 107--112.
\item Dubrovin, B. A., Fomenko, A. T. \& Novikov, S. P. 1984
  \textit{Modern geometry -- methods and applications. Part I. The geometry of
  surfaces, transformation groups, and fields}, New York: Springer Verlag.
\item Godt, W. 1896 Ueber eine merkw\"urdige Kreisfigur. \textit{Math. 
  Ann.} \textbf{47}, 564--572.
\item Grace, J. H. 1898 Circles, spheres and linear complexes. \textit{Trans. 
  Camb. Phil. Soc.} \textbf{16}, 153--190.
\item Hirota, R. 1981 Discrete analogue of a generalized Toda equation.
  \textit{J. Phys. Soc. Japan} \textbf{50}, 3785--3791.
\item Koecher, M. \& Remmert, R. 1991 Hamilton's 
  quaternions. In \textit{Numbers}  (ed. H.-D. Ebbinghaus, H. Hermes, 
  F. Hirzebruch, M. Koecher, K. Mainzer, J. Neukirch, A. Prestel \& R. Remmert).
  Graduate Text in Mathematics/Readings in Mathematics, no. 123, pp. 189--220, 
  New York: Springer Verlag.
\item Konopelchenko, B. G. \& Schief, W. K. 2002  Menelaus' theorem, Clifford 
  configurations and inversive geometry of the Schwarzian KP hierarchy. 
  \textit{J. Phys. A: Math. Gen.} \textbf{35}, 6125--6144.
\item Konopelchenko B. G. \& Schief, W. K. 2005 Conformal geometry of the 
  (discrete) Schwarzian Davey-Stewartson II hierarchy. \textit{Glasgow Math. 
  J.} \textbf{47A}, 121--131.
\item Longuet-Higgins, M. S. 1972 Clifford's chain and its
  analogues in relation to the higher polytopes. \textit{Proc. R. Soc. Lond.} A
  \textbf{330}, 443--466.
\item Pedoe, D. 1970 \textit{A course of geometry}, Cambridge University
  Press.
\item Zakharov, V. E., Manakov, S. V., Novikov, S. P. \& Pitaevski, L. P.
  1980 \textit{Soliton theory. The inverse problem method}, Moscow: Nauka;
  1984 New York: Plenum Press. 
\item Ziegenbein, P. 1941 Konfigurationen in der Kreisgeometrie. 
  \textit{J. reine angew. Math.} \textbf{183}, 9--24.
\end{thedemobiblio}

\end{document}